\documentclass[onecolumn,12pt]{IEEEtran} 
\usepackage{color}
\usepackage{amsfonts,latexsym}
\usepackage{amssymb}

\newcommand{\Z}{{\mathbf{Z}}}
\newcommand{\fS}{{\mathcal{S}}}
\newcommand{\cS}{{\mathcal{S}}}

\newcommand{\lL}{{\mathbf L}}
\newcommand{\mM}{{\mathbf M}}

\newcommand{\C}{{\mathbb C}}
\newcommand{\hH}{{\mathbb H}}

\newtheorem{theorem}{Theorem}
\newtheorem{lemma}[theorem]{Lemma}

\newtheorem{remark}{Remark}

\def\BibTeX{{\rm B\kern-.05em{\sc i\kern-.025em b}\kern-.08em
    T\kern-.1667em\lower.7ex\hbox{E}\kern-.125emX}}

\begin{document}

\title{Bounds on and Constructions of Unit Time-Phase Signal Sets}

\author{Cunsheng Ding,\thanks{
C. Ding is with the Department of Computer Science and Engineering,
The Hong Kong University of Science and Technology, Clear Water Bay,
Kowloon, Hong Kong. Email: cding@ust.hk}
Keqin Feng,\thanks{K. Feng is with the Department of Mathematical Sciences, Tsinghua University,
 Beijing, 100084, P. R. China. Email: kfeng@math.tsinghua.edu.cn}
Rongquan Feng, \thanks{R. Feng is with the School of Mathematical
Sciences, Peking University, Beijing 100871, P. R. China. Email:
fengrq@math.pku.edu.cn} Aixian Zhang \thanks{A. Zhang is with the
School of Mathematical Sciences, Capital Normal University,
Beijing,100048, P. R. China. Email: zhangaixian1008@126.com}
}

\date{\today}
\maketitle

\begin{abstract}
Digital signals are complex-valued functions on $\Z_n$. Signal sets
with certain properties are required in various communication
systems. Traditional signal sets consider only the time distortion
during transmission. Recently, signal sets against both the time and
phase distortion have been studied, and are called {\em time-phase}
signal sets. Several constructions of time-phase signal sets are
available in the literature. There are a number of bounds on time
signal sets (also called codebooks). They are automatically bounds on time-phase signal sets,
but are bad bounds. The first objective of this paper is to develop better
bounds on time-phase signal sets from known bounds on time signal
sets. The second objective of this paper is to construct two series of
time-phase signal sets, one of which is optimal.
\end{abstract}


\section{Introduction}

Throuout this paper, let $n>1$ be an integer, and let
$\Z_n=\{0,1,\cdots, n-1\}$. We define $\hH_n=\C(\Z_n)$, the set of
all complex-valued functions on $\Z_n$, which is a Hilbert space
with the Hermitian product given by
$$
\langle \phi, \varphi \rangle =\sum_{t \in \Z_n} \phi(t) \overline{\varphi(t)}.
$$
The {\em norm} of $\phi \in \hH_n$ is defined by
$$
\| \phi \| = \sqrt{\langle \phi, \phi \rangle}.
$$

Digital signals are complex-valued functions on $\Z_n$. They are
also called sequences as the following mapping
$$
\phi \mapsto  (\phi(0), \phi(1), \cdots, \phi(n-1))
$$
transfers a function $\phi \in \hH_n$ into a sequence in $\C^n$. We
identify the function $\phi$ with its sequence. A subset $\fS
\subset \hH_n$ is called a {\em signal set}, and a {\em real signal
set} if every signal in $\fS$ is real-valued. Any $\phi \in \hH_n$
is called a {\em unit signal} if $\| \phi \|=1$. A subset $\fS$ is
said to be a unit signal set if every signal in $\fS$ is a unit
signal. Signal sets with certain properties are required in some
communication systems \cite{GHS,HCM,WG11}.

During the transmission process, a signal $\varphi$ might be distorted.
Two basic types of distortion are the {\em time shift}
$\varphi(t) \mapsto \lL_\tau \varphi(t)=\varphi(t+\tau)$ and the {\em phase shift}
$\varphi(t) \mapsto \mM_w \varphi(t)=e^{\frac{2\pi i}{n}wt} \varphi(t)$,
where $\tau, w \in \Z_n$. For certain applications, it is required that for every $\varphi \ne \phi \in \fS$,
\begin{eqnarray}\label{eqn-brequirment}
|\langle \phi, \mM_w \lL_\tau \varphi \rangle |  \ll 1.
\end{eqnarray}
In addition, signals are sometimes required to admit low peak-to-average power ratio,
i.e., for every $\phi \in \fS$ with $\| \phi \|=1$,
$$
\max\{|\phi(t)|: t \in \Z_n\} \ll 1.
$$
In view of the above requirements and that every signal can be normalized into a unit signal,
in this paper we consider only {\em unit signal sets} $\fS$, in which we have for every unit
signal $\phi$
$$
\frac{1}{\sqrt{n}} \le \max\{|\phi(t)|: t \in \Z_n\}.
$$

To measure the capability of anti-distortion of a signal set $\fS$ with respect
to the time and phase shift, Gurevich, Hadani and Sochen \cite{GHS} defined 
\begin{eqnarray}\label{eqb-correlation1}
\lambda_{\fS}=\max\{|\langle \phi, \mM_w \lL_\tau \varphi \rangle |:
                     \mbox{ either } \phi \ne \varphi \mbox{ or } (\tau, w) \ne (0,0) \}.
\end{eqnarray}

For convenience, we call $\fS$ an $(n, M, \lambda_\fS)$ {\em time-phase signal set},
where $M$ denotes the total number of signals in $\fS$. If we require that $\lambda_\fS<1$,
any  $(n, M, \lambda_\fS)$ time-phase signal set is an {\em ambiguity signal set} defined in
\cite{WG11}.

In some communication systems, only time distortion is considered. In this case, two correlation
measures are considered.  One is the maximum crosscorrelation amplitude $\nu_\cS$ of
an $(n, M)$ signal set defined by
\begin{eqnarray}\label{eqn-correlation2}
\nu_\cS = \max_{\phi, \varphi \in \fS \atop \phi \ne \varphi} |\langle \phi, \varphi \rangle |,
\end{eqnarray}
and the other is the maximum auto-and-cross correlation  amplitude $\theta_\cS$  of
an $(n, M)$ signal set defined by
\begin{eqnarray}\label{eqb-correlation3}
\theta_{\fS}=\max\{|\langle \phi, \lL_\tau \varphi \rangle |:
                     \mbox{ either } \phi \ne \varphi \mbox{ or } \tau \ne 0 \}.
\end{eqnarray}
If only time distortion is considered, we call $\fS$ an $(n, M, \nu_\cS)$ or $(n, M, \theta_\cS)$
{\it time signal set}. Time signal sets are also called {\em codebooks}.

By definition, we have clearly
\begin{eqnarray}\label{eqn-boundchain}
\lambda_\cS \ge \theta_\cS \ge \nu_\cS
\end{eqnarray}
for any $(n, M)$ signal set.

Hence, signal sets are classified into two types: time-phase signal sets and time signal sets.
Time signal sets have been studied for CDMA communications (see, for example, \cite{Allop,CCKS,CHS,Ding05,DYin,HSAP,HCM,LHS,Sarw98,SH03,Welch, HKchapter}). A number of
lower bounds on $\theta_\cS$ and $\nu_\cS$ were developed. They are automatically
lower bounds on $\lambda_\cS$ due to (\ref{eqn-boundchain}), but are bad lower bounds
for $\lambda_\cS$ as the correlation measure $\lambda_\cS$ is much stronger than $\theta_\cS$
and $\nu_\cS$.
The objectives of this paper are to derive better lower bounds on the parameters of
time-phase signal sets, and construct optimal and good optimal time-phase signal sets.

This paper is organized as follows. Section \ref{sec-ss} first establishes a one-way
bridge between time-phase signal sets and time signal sets, and then uses this bridge to
develop bounds on unit time-phase signal sets from known bounds on time signal sets.
Section \ref{sec-2ndgroup} first sets up another one-way bridge
between time-phase signal sets and time signal sets, and then employs this bridge to
develop more bounds on unit time-phase signal sets from known bounds on time signal sets.
Section \ref{sec-constructions} presents two series of unit time-phase signal sets, one 
of which is optimal.
Section \ref{sec-summary} summaries the main contributions of this paper and presents
some open problems.

\section{The first group of lower bounds on $(n, M, \lambda)$ time-phase signal sets}\label{sec-ss}

Two bounds on the parameters of $(n, M, \nu_\cS)$ time signal sets are described
in the following two lemmas.

\begin{lemma}\label{lem-Welch}
{\rm (Welch's bound \cite{Welch})} For any $(n, M, \nu)$ unit time
signal set  $\cS$ with $M  \geq n$, and for each integer $k \geq 1,$
$$
\sum_{\phi, \varphi \in
\cS} | \langle \phi, \varphi \rangle |^{2k} \geq  { n+k-1 \choose k
}^{-1} M^{2},
$$ so that
\begin{eqnarray}\label{eqn-welchb2}
\nu_\cS &\ge& \left(\frac{1}{M(M-1)} \sum_{\scriptstyle \phi, \varphi \in
\cS \atop \scriptstyle \phi \neq \varphi} |\langle \phi, \varphi
\rangle|^{2k} \right)^{1/2k} \nonumber \\
&\geq & \left(\frac{{ n+k-1 \choose k }
^{-1}M^{2}-M}{M(M-1) }\right)^{1/2k}  \nonumber \\
&=& \left(\frac{M-{ n+k-1 \choose k
}}{(M-1){ n+k-1 \choose k }}\right)^{1/2k} \triangleq  \widetilde{w_{k}},
\end{eqnarray}
and $\nu_\cS  = \widetilde{w_{k}} $ if and only if for all pairs
$(\phi, \varphi) \in \cS \times \cS $ with $ \phi \ne \varphi,
|\langle \phi, \varphi \rangle |=\widetilde{w_{k}}.$
\end{lemma}

It was proved in \cite{SH03} that
no $(n,M, \nu)$ real time signal set  $\cS$ can meet the Welch bound
$\widetilde{w_{1}} $ of (\ref{eqn-welchb2}) if $M > n(n+1)/2$ and
no $(n,M,\nu)$ time signal set $\cS$ can meet the Welch bound
$\widetilde{w_{1}} $ of (\ref{eqn-welchb2}) if $M > n^2$.

The following was proved in \cite{Levens}

\begin{lemma}\label{lem-Levenstein}
For any $(n,M,\nu_\cS)$ real unit time signal set $\cS$ with $M > n(n+1)/2$,
\begin{eqnarray}\label{eqn-boundKL1}
\nu_\cS \ge \sqrt{\frac{3M-n^2-2n}{(n+2)(M-n)}}.
\end{eqnarray}

For any $(n,M,\nu_\cS)$ complex unit time signal set $\cS$ with $M > n^2$,
\begin{eqnarray}\label{eqn-boundKL2}
\nu_\cS \ge \sqrt{\frac{2M-n^2-n}{(n+1)(M-n)}}.
\end{eqnarray}
\end{lemma}

If $M<n(n+1)/2$ (respectively $M<n^2$), the Welch bound
$\widetilde{w_{1}}=\sqrt{\frac{M-n}{(M-1)n}}$ on real (respectively,
complex) time signal sets is better. However, the Levenstein bound of
(\ref{eqn-boundKL1}) on real time signal sets is tighter than Welch's
bound $\widetilde{w_{1}}$ if $M>n(n+1)/2$, and that the Levenstein
bound of (\ref{eqn-boundKL2}) on complex time signal sets is tighter than
Welch's bound $\widetilde{w_{1}}$ if $M>n^2$.

Welch's and Levenstein's lower bounds on $\nu$ for time signal sets
yield directly lower bounds on $\lambda$ for time-phase signal sets
according to (\ref{eqn-boundchain}). But they are very bad lower
bounds on $\lambda$ as the correlation measure $\lambda_\cS$ is much
stronger than $\nu_\cS$ and $\theta_\cS$. However, they can be employed
to derive better bounds on $\lambda_\cS$ for time-phase signal sets. This
is our task in this section.

\begin{lemma}\label{lem-nnew1}
For any pair of distinct signals $\phi$ and $\varphi$ in an $(n, M, \lambda)$ unit
time-phase signal set with $\lambda<1$, we have
$$
\phi \ne \mM_w \lL_{\tau} \varphi
$$
for any $w, \tau \in \Z_n$.
\end{lemma}

\begin{proof}
The conclusion follows directly from the assumption that $\lambda<1$.
\end{proof}

Given an $(n, M, \lambda)$ unit time-phase signal set $\fS$ with $\lambda<1$, where 
$$
\fS=\{\phi_1, \phi_2, \cdots, \phi_M\},
$$
we  define an $(n, n^2M, \lambda)$ unit time signal set
\begin{eqnarray}\label{eqn-mainS}
\cS_{\fS}=\left\{ \phi_{j,w,\tau}: 1 \le j \le M, \ 0 \le w \le n-1, \ 0 \le \tau \le n-1  \right\},
\end{eqnarray}
where
\begin{eqnarray}\label{eqn-mainS2}
\phi_{j,w,\tau}(t)=e^{\frac{2\pi i}{n}wt}\phi_j(t+\tau).
\end{eqnarray}

\begin{lemma}\label{lem-nnew2}
Let $\phi_{j,w,\tau}$ be defined in (\ref{eqn-mainS2}).
Then each $\phi_{j,w,\tau}$ is a unit signal. In addition $\phi_{j,w,\tau}=\phi_{j',w',\tau'}$ if and only if
           $(j, w, \tau)=(j', w', \tau')$.
\end{lemma}

\begin{proof}
The first conclusion is straightforward, and the second follows from Lemma \ref{lem-nnew1}.
\end{proof}

\begin{theorem}\label{thm-mmain}
For any $(n, M, \lambda_\fS)$ unit time-phase signal set $\fS$ with $\lambda_\cS<1$,
the set $\cS_{\fS}$ defined in (\ref{eqn-mainS}) is an $(n, n^2M, \nu_{\cS_{\fS}})$
unit time signal set with $\nu_{\cS_{\fS}}=\lambda_\cS$.
\end{theorem}

\begin{proof}
It follows from Lemma \ref{lem-nnew2} that $\left|\cS_{\fS}\right|=n^2M$.
By definition,
\begin{eqnarray*}
|\langle \phi_{j,w,\tau}, \phi_{j',w',\tau'} \rangle|
&=& \left| \sum_{t=0}^{n-1} e^{\frac{2\pi i}{n}wt}\phi_j(t+\tau)
                            e^{-\frac{2\pi i}{n}w't}\overline{\phi_{j'}(t+\tau')}   \right| \\
&=& \left| \sum_{t=0}^{n-1} \phi_j(t+\tau)
                            e^{-\frac{2\pi i}{n}(w'-w)t}\overline{\phi_{j'}(t+\tau')}   \right| \\
&=& \left| \sum_{t=0}^{n-1} \phi_j(t)
                            e^{-\frac{2\pi i}{n}(w'-w)(t-\tau)}\overline{\phi_{j'}(t+\tau'-\tau)}   \right| \\
&=& \left| \sum_{t=0}^{n-1} \phi_j(t)
           e^{-\frac{2\pi i}{n}(w'-w)t}\overline{\phi_{j'}(t+\tau'-\tau)}   \right| \\
&=& \left| \sum_{t=0}^{n-1} \phi_j(t)
           \overline{e^{\frac{2\pi i}{n}(w'-w)t} \phi_{j'}(t+\tau'-\tau)}   \right| \\
&=& \left| \langle \phi_j, \mM_{(w'-w) \bmod{n}} \lL_{(\tau'-\tau) \bmod{n}} \phi_{j'} \rangle   \right|.
\end{eqnarray*}

Note that $(w'-w) \bmod{n}$ ranges over all elements in $\Z_n$ when both $w$ and $w'$ run
over all elements in $\Z_n$. The same is true for $(\tau'-\tau) \bmod{n}$. Hence, the
signal set $\cS_{\fS}$ has maximum crosscorrelation amplitude $\nu_{\fS_\fS}=\lambda_\fS$.
\end{proof}

\begin{remark}
Theorem \ref{thm-mmain} is one of the main contributions of this paper.
It serves as a one-way bridge with which bounds on time signal sets can be
employed to derive better lower bounds on time-phase signal sets.
\end{remark}

\begin{theorem}\label{thm-Welch-DFF}
For any $(n, M, \lambda_\fS)$ unit time-phase signal set $\fS$ with
$\lambda_\fS<1$ and each integer $k \geq 1,$ we have
\begin{eqnarray}\label{eqn-welchb2-DFF}
\lambda_\fS \ge w_{k}  \triangleq \left(\frac{ n^{2} M-{ n+k-1 \choose k
}}{(n^{2}M-1){ n+k-1 \choose k }}\right)^{1/2k}.
\end{eqnarray}
\end{theorem}

\begin{proof}
The desired conclusion follows from Lemma \ref{lem-Welch} and Theorem \ref{thm-mmain}.
\end{proof}

\begin{theorem}\label{thm-LDFF2}
For any $(n, M, \lambda_\fS)$ unit time-phase signal set $\fS$ with $\lambda<1$ and $M>1$,
\begin{eqnarray}\label{eqn-LDFF2}
\lambda_\fS \ge \sqrt{\frac{2nM-n-1}{(n+1)(nM-1)}}.
\end{eqnarray}
\end{theorem}

\begin{proof}
The desired conclusion follows from Lemma \ref{lem-Levenstein} and Theorem \ref{thm-mmain}.
\end{proof}

\begin{remark}
For any $(n,M, \lambda)$ unit time-phase signal set $\fS$ with $M>n^2$, the Levenstein bound
gives automatically the following bound:
$$
\lambda_\cS \ge \sqrt{\frac{2M-n^2-n}{(n+1)(M-n)}}.
$$
However, it can be checked that this lower bound is much smaller than the lower
bound of (\ref{eqn-LDFF2}) when $M>n^2$. This shows that the new bound of
(\ref{eqn-LDFF2}) is indeed a big improvement over the original Levenstein bound
when it is used as a bound for unit time-phase signal sets.
This comment also applies to the new bound of Theorem \ref{thm-Welch-DFF}.
\end{remark}

\begin{remark}
It is easy to verify that the bounds $w_{1}$ of
(\ref{eqn-welchb2-DFF}) and (\ref{eqn-LDFF2}) are the same when
$M=1$, and the latter is superior when $M>1$. So the bound of
(\ref{eqn-welchb2-DFF}) is useful only for the case that $M=1$. The
purpose of presenting Theorem \ref{thm-Welch-DFF} here is to show
that the Levenstein bound yields better bound on unit time-phase
signal sets under the framework of this section. We will construct
some series of unit time-phase signal sets in Section $IV$ and then
discuss the tightness of these bounds in Section $V$.
\end{remark}

\section{The second group of bounds on $(n, M, \lambda)$ time-phase signal sets}\label{sec-2ndgroup}

In this section, we further derive bounds on time-phase unit signal sets from
known bounds on time unit signal sets.

Given an $(n, M, \lambda_\fS)$  time-phase unit signal set $\fS$ with $\lambda_\fS<1$,
where
$$
\fS=\{\phi_1, \phi_2, \cdots, \phi_M\},
$$
we now define an $(n, nM)$ unit time signal set
\begin{eqnarray}\label{eqn-mainS211}
\bar{\cS}_{\fS}=\left\{ \phi_{j,w}: 1 \le j \le M, \ 0 \le w \le n-1  \right\},
\end{eqnarray}
where
\begin{eqnarray}\label{eqn-mainS22}
\phi_{j,w}(t)=e^{\frac{2\pi i}{n}wt}\phi_j(t).
\end{eqnarray}

\begin{lemma}\label{lem-nnew22}
Let $\phi_{j,w}$ be defined in (\ref{eqn-mainS22}).
Then each $\phi_{j,w}$ is a unit signal.
In addition,  $\phi_{j,w}=\phi_{j',w'}$ if and only if
           $(j, w)=(j', w')$.
\end{lemma}

\begin{proof}
The first conclusion is straightforward, and the second one follows from
Lemma \ref{lem-nnew1}.
\end{proof}

\begin{theorem}\label{thm-mmain2}
For any $(n, M, \lambda_\fS)$ unit time-phase signal set $\fS$ with $\lambda_\fS<1$,
the set $\bar{\cS}_{\fS}$ defined in (\ref{eqn-mainS211}) is an $(n, nM, \theta_{\bar{\cS}_{\fS}})$
unit time signal set with $\theta_{\bar{\cS}_{\fS}}=\lambda_\cS$.
\end{theorem}

\begin{proof}
It follows from Lemma \ref{lem-nnew22} that $\left|\bar{\cS}_{\fS}\right|=nM$.
By definition,
\begin{eqnarray*}
|\langle \phi_{j,w}, \lL_\tau\phi_{j',w'} \rangle|
&=& \left| \sum_{t=0}^{n-1} e^{\frac{2\pi i}{n}wt}\phi_j(t)
                            e^{-\frac{2\pi i}{n}w'(t+\tau)}\overline{\phi_{j'}(t+\tau)}   \right| \\
&=& \left| \sum_{t=0}^{n-1} \phi_j(t)
                            e^{-\frac{2\pi i}{n}(w'-w)t}\overline{\phi_{j'}(t+\tau)}   \right| \\
&=& \left| \sum_{t=0}^{n-1} \phi_j(t)
           \overline{e^{\frac{2\pi i}{n}(w'-w)t} \phi_{j'}(t+\tau)}   \right| \\
&=& \left| \langle \phi_j, \mM_{(w'-w) \bmod{n}} \lL_{\tau} \phi_{j'} \rangle   \right|.
\end{eqnarray*}

Note that $(w'-w) \bmod{n}$ ranges over all elements in $\Z_n$ when both $w$ and $w'$ run
over all elements in $\Z_n$.  Hence, for the
signal set $\bar{\cS}_{\fS}$ we have  $\theta_{\bar{\cS}_{\fS}} = \lambda_\fS$.
\end{proof}

\begin{remark}
Theorem \ref{thm-mmain2} is another main contribution of this paper. It
serves as a one-way bridge with which bounds on time signal sets can be
employed to derive better lower bounds on time-phase signal sets.
\end{remark}

The following bounds on unit time signal sets are due to Levenstein \cite{Levens,HKchapter},
and are linear programming bounds.
They work automatically as bounds for unit time-phase signal sets, but are bad ones
because the bounds of Theorem \ref{thm-complex} are much better.

\begin{lemma}\label{lem-complex}
Let $\cS \subset \hH_n$ be any $(n, M, \theta)$ unit time signal set. Then
\begin{eqnarray*}
M \le \left\{
\begin{array}{ll}
\frac{1-\theta^2}{1-n\theta^2}, \ &\mbox{ if }  0 < \theta^2 \le \frac{1}{n+1}  \\
\frac{(n+1)(1-\theta^2)}{2-(n+1)\theta^2}, \ &\mbox{ if }  \frac{1}{n+1} < \theta^2 \le \frac{2}{n+2}  \\
\frac{n(n+1)(n+2)(1-\theta^2)^2}{(n+1)(n+2)\theta^4 -4(n+1)\theta^2 +2 },
                   \ &\mbox{ if }  \frac{2}{n+2} < \theta^2 \le \frac{2(n+2)+\sqrt{2(n+1)(n+2)}}{(n+2)(n+3)}  \\
\frac{n(n+1)(n+2)[(n+3)\theta^2-2](1-\theta^2)}{12(n+2)\theta^2 -2(n+2)(n+3)\theta^4 -12 },
       \ & \mbox{ if }  \frac{2(n+2)+\sqrt{2(n+1)(n+2)}}{(n+2)(n+3)} \le \theta^2 \le \frac{3(n+3)+\sqrt{3(n+3)(n+1)}}{(n+3)(n+4)}.
\end{array}
\right.
\end{eqnarray*}
\end{lemma}

The bounds on unit time-phase signal sets described in the following theorem are derived from the
 bounds of Lemma \ref{lem-complex} and are better.

\begin{theorem}\label{thm-complex}
Let $\fS \subset \hH_n$ be any $(n, M, \lambda)$ unit time-phase signal set. Then
\begin{eqnarray}\label{eqn-baby0}
nM \le \left\{
\begin{array}{ll}
\frac{1-\lambda^2}{1-n\lambda^2}, \ &\mbox{ if }  0 < \lambda^2 \le \frac{1}{n+1}  \\
\frac{(n+1)(1-\lambda^2)}{2-(n+1)\lambda^2}, \ &\mbox{ if }  \frac{1}{n+1} < \lambda^2 \le \frac{2}{n+2}  \\
\frac{n(n+1)(n+2)(1-\lambda^2)^2}{(n+1)(n+2)\lambda^4 -4(n+1)\lambda^2 +2 },
                   \ &\mbox{ if }  \frac{2}{n+2} < \lambda^2 \le \frac{2(n+2)+\sqrt{2(n+1)(n+2)}}{(n+2)(n+3)}  \\
\frac{n(n+1)(n+2)[(n+3)\lambda^2-2](1-\lambda^2)}{12(n+2)\lambda^2 -2(n+2)(n+3)\lambda^4 -12 },
       \ &\mbox{ if }  \frac{2(n+2)+\sqrt{2(n+1)(n+2)}}{(n+2)(n+3)} \le \lambda^2 \le \frac{3(n+3)+\sqrt{3(n+3)(n+1)}}{(n+3)(n+4)}.
\end{array}
\right.
\end{eqnarray}
\end{theorem}

\begin{proof}
The desired conclusions of this theorem follow from Theorem \ref{thm-mmain2} and Lemma \ref{lem-complex}.
\end{proof}

\begin{remark}
The first bound in (\ref{eqn-baby0}) coincides with the bound of (\ref{eqn-welchb2-DFF})
in the case $k=1$. The second bound in (\ref{eqn-baby0}) coincides with the bound of (\ref{eqn-LDFF2}).
So these bounds can be derived with both bridges established in this paper.

\end{remark}

To introduce more bounds on time-phase unit signal sets, let $\hH_{(n,q)}$ denote
the set of all complex-valued functions $f$ on $\Z_n$ such that $\sqrt{n}f(i)$ is a $q$th root
of unity for all $i \in \Z_n$.

The following bounds on unit time signal sets are due to Levenstein \cite{Levens,HKchapter},
and are linear programming bounds.
They work automatically as bounds for unit time-phase signal sets, but are bad ones
because the bounds of Theorem \ref{thm-complexq=2} are much better.

\begin{lemma}\label{lem-complexq=2}
Let $\cS \subset \hH_{(n,q)}$ be any $(n, M, \theta)$ unit time  signal set, where $q=2$. Then
\begin{eqnarray*}
M \le \left\{
\begin{array}{ll}
\frac{1-\theta^2}{1-n\theta^2}, \ &\mbox{ if }  0 \le \theta^2 \le \frac{n-2}{n^2}  \\
\frac{n^2(1-\theta^2)}{3n-2-n^2 \theta^2}, \ &\mbox{ if }  \frac{n-2}{n^2} \le \theta^2 \le \frac{3n-8}{n^2}  \\
\frac{n(1-\theta^2)[(n-2)(n^2-3n+8)-(n^2-n+2)n^2\theta^2  ]}{6n(n-2)-4(3n-4)n^2\theta^2+2n^4\theta^4},  & \mbox{ if }  \frac{3n-8}{n^2} \le \theta^2 \le \frac{3n-10+\sqrt{6n^2-42n+76}}{n^2}  \\
\frac{n^2(1-\theta^2)}{6}
\frac{3n^3-23n^2+90n-136-(n^2-3n+8)n^2\theta^2}{15n^2-50n+24-10(n-2)n^2\theta^2+n^4\theta^4},
& \mbox{ if } \frac{3n-10+\sqrt{6n^2-42n+76}}{n^2}    \le \theta^2
\le \frac{5(n-4)+\sqrt{10n^2-90n+216}}{n^2}.
\end{array}
\right.
\end{eqnarray*}
\end{lemma}

The bounds on unit time-phase signal sets described in the following theorem are derived from the
 bounds of Lemma \ref{lem-complexq=2} and are better.

\begin{theorem}\label{thm-complexq=2}
Let $\fS \subset \hH_{(n,q)}$ be any $(n, M, \lambda)$ unit time-phase signal set, where $q=2$. Then
\begin{eqnarray}\label{eqn-babyq=2}
nM \le \left\{
\begin{array}{ll}
\frac{1-\lambda^2}{1-n\lambda^2}, &\mbox{if }  0 \le \lambda^2 \le \frac{n-2}{n^2}  \\
\frac{n^2(1-\lambda^2)}{3n-2-n^2 \lambda^2}  &\mbox{if }  \frac{n-2}{n^2} \le \lambda^2 \le \frac{3n-8}{n^2}  \\
\frac{n(1-\lambda^2)[(n-2)(n^2-3n+8)-(n^2-n+2)n^2\theta^2  ]}{6n(n-2)-4(3n-4)n^2\theta^2+2n^4\theta^4}  & \mbox{if }  \frac{3n-8}{n^2} \le \lambda^2 \le \frac{3n-10+\sqrt{6n^2-42n+76}}{n^2}  \\
\frac{n^2(1-\lambda^2)}{6}
\frac{3n^3-23n^2+90n-136-(n^2-3n+8)n^2\theta^2}{15n^2-50n+24-10(n-2)n^2\theta^2+n^4\theta^4} 
& \mbox{if } \frac{3n-10+\sqrt{6n^2-42n+76}}{n^2}    \le \lambda^2
\le \frac{5(n-4)+\sqrt{10n^2-90n+216}}{n^2}.
\end{array}
\right.
\end{eqnarray}
\end{theorem}

\begin{proof}
The desired conclusions of this theorem follow from Theorem \ref{thm-mmain2} and Lemma \ref{lem-complexq=2}.
\end{proof}

\begin{remark}
The second bound in (\ref{eqn-babyq=2}) is better than the second bound in (\ref{eqn-baby0}) when
$n>2$. But the former applies only to binary real time-phase signal sets, while the latter applies to all time-phase signal sets.
\end{remark}

The following bounds on unit time signal sets are due to Levenstein \cite{Levens,HKchapter},
and are linear programming bounds.
They work automatically as bounds for unit time-phase signal sets, but are bad ones
because the bounds of Theorem \ref{thm-complexq>2} are much better.

\begin{lemma}\label{lem-complexq>2}
Let $\cS \subset \hH_{(n,q)}$ be any $(n, M, \theta)$ unit time  signal set, where $q \ge 3$. Then
\begin{eqnarray*}
M \le \left\{
\begin{array}{ll}
\frac{1-\theta^2}{1-n\theta^2}, \ &\mbox{ if }  0 \le \theta^2 \le \frac{n-1}{n^2}  \\
\frac{n^2(1-\theta^2)}{2n-1-n^2 \theta^2}, \ &\mbox{ if }  \frac{n-1}{n^2} \le \theta^2 \le \frac{2n^2-5n+4}{n^2(n-1)}  \\
\frac{n^2(1-\theta^2)[(n^2-n+1)n^2\theta^2 -n^3+3n^2-5n+4
]}{n[4(n-1)n^2\theta^2 - n^4 \theta^4-2n^2+3n]},  &  \mbox{ if }
\frac{2n^2-5n+4}{n^2(n-1)}    \le \theta^2 \le
\frac{2n-2+\sqrt{2n^2-5n+4}}{n^2}.
\end{array}
\right.
\end{eqnarray*}
\end{lemma}

The bounds on unit time-phase signal sets described in the following theorem are derived from the
 bounds of Lemma \ref{lem-complexq>2} and are better.

\begin{theorem}\label{thm-complexq>2}
Let $\fS \subset \hH_{(n,q)}$ be any $(n, M, \lambda)$ unit time-phase signal set, where $q \ge 3$. Then
\begin{eqnarray*}
nM \le \left\{
\begin{array}{ll}
\frac{1-\lambda^2}{1-n\lambda^2}, \ &\mbox{ if }  0 \le \lambda^2 \le \frac{n-1}{n^2}  \\
\frac{n^2(1-\lambda^2)}{2n-1-n^2 \lambda^2}, \ &\mbox{ if }  \frac{n-1}{n^2} \le \lambda^2 \le \frac{2n^2-5n+4}{n^2(n-1)}  \\
\frac{n^2(1-\lambda^2)[(n^2-n+1)n^2\theta^2 -n^3+3n^2-5n+4
]}{n[4(n-1)n^2\theta^2 - n^4 \lambda^4-2n^2+3n]},  & \mbox{ if }
\frac{2n^2-5n+4}{n^2(n-1)}    \le \lambda^2 \le
\frac{2n-2+\sqrt{2n^2-5n+4}}{n^2}.
\end{array}
\right.
\end{eqnarray*}
\end{theorem}

\begin{proof}
The desired conclusions of this theorem follow from Theorem \ref{thm-mmain2} and Lemma \ref{lem-complexq>2}.
\end{proof}

The following bounds on unit time signal sets are due to Sidelnikov \cite{Sidel}.
They work automatically as bounds for unit time-phase signal sets, but are bad ones
because the bounds of Theorem \ref{thm-Sidel} are much better.

\begin{lemma}\label{lem-Sidel}
Let $\cS \subset \hH_{(n,q)}$ be any $(n, M, \theta)$ unit time signal set. Then
\begin{eqnarray*}
\theta^2 \ge \left\{
\begin{array}{ll}
\frac{(2k+1)(n-k)}{n^2}+\frac{k(k+1)}{2n^2} - \frac{2^k n^{2k}}{M(2k)! {n \choose k}}
     & \mbox{ if } 0 \le k \le \frac{2n}{5} \mbox{ and } q=2 \\
\frac{(k+1)(2n-k)}{2n^2} - \frac{2^k n^{2k}}{M(k!)^2 {n \choose k}}
     & \mbox{ if } k \ge 0 \mbox{ and } q>2.
\end{array}
\right.
\end{eqnarray*}
\end{lemma}

The bounds on unit time-phase signal sets described in the following theorem are derived from the
Sidelnikov bounds of Lemma \ref{lem-Sidel} and are better.

\begin{theorem}\label{thm-Sidel}
Let $\cS \subset \hH_{(n,q)}$ be any $(n, M, \lambda)$ unit time-phase signal set. Then
\begin{eqnarray*}
\lambda^2 \ge \left\{
\begin{array}{ll}
\frac{(2k+1)(n-k)}{n^2}+\frac{k(k+1)}{2n^2} - \frac{2^k n^{2k}}{nM(2k)! {n \choose k}}
     & \mbox{ if } 0 \le k \le \frac{2n}{5} \mbox{ and } q=2 \\
\frac{(k+1)(2n-k)}{2n^2} - \frac{2^k n^{2k}}{nM(k!)^2 {n \choose k}}
     & \mbox{ if } k \ge 0 \mbox{ and } q>2.
\end{array}
\right.
\end{eqnarray*}
\end{theorem}

\begin{proof}
The desired conclusions of this theorem follow from Theorem \ref{thm-mmain2} and Lemma \ref{lem-Sidel}.
\end{proof}

\section{Constructions of unit time-phase signal sets}\label{sec-constructions}

In this section we present two series of unit time-phase signal
sets which are related to Gaussian sums. We first introduce some
basic facts on Gaussian sums that will be employed in this section. 
For more information
the reader is referred to \cite{BEW}.

Let $\zeta_{n}=e^{\frac{2\pi\sqrt{-1}}{n}} \ (n \geq 2), \ q=p^{m}$
where $m \ge 1$ and $p$ is a prime number. Let $ \mathrm{T}:
\mathbb{F}_{q} \rightarrow \mathbb{F}_{p}$ be the trace mapping. The
group of additive characters of $(\mathbb{F}_{q}, \ +)$ is
$$\widehat{\mathbb{F}}_{q}= \{ \psi_{b}: b \in \mathbb{F}_{q} \},$$
where $\psi_{b}: \mathbb{F}_{q} \rightarrow \langle \zeta_{p}
\rangle$ is defined by
\begin{eqnarray}\label{eqn-AC}
\psi_{b}(x)=\zeta_{p}^{\mathrm{T}(bx)} \ \ (x \in \mathbb{F}_{q}
).
\end{eqnarray}
The identity (trivial character ) is $\psi_{0}=1$ and the inverse
of $\psi_{b}$ is $\psi_{-b}.$

Let $\gamma$ be a primitive element of $\mathbb{F}_{q}$ so that
$$
\mathbb{F}^{\times}_{q}=\mathbb{F}_{q} \backslash
\{0\}=\{1,\gamma,\gamma^{2},\cdots,\gamma^{q-2}\}.
$$
The group of multiplicative characters of $\mathbb{F}_{q}$ is
$$
(\mathbb{F}^{\times}_{q})^{\wedge}=\langle \omega  \rangle
=\{\omega^{i}: 0 \leq i \leq q-2 \},
$$
where $\omega: \mathbb{F}^{\times}_{q} \rightarrow \langle
\zeta_{q-1} \rangle $ is defined by
$$
\omega(\gamma^{j})=\zeta^{j}_{q-1} \ \ (0  \leq j \leq q-2).
$$
The identity (trivial character ) is $\omega^{0}=1$ and the inverse
of $\omega^{i}$ is the conjugate character
$\bar{\omega}^{i}=\omega^{-i}.$

For an additive character $\psi$ and multiplicative character
$\chi$ of $\mathbb{F}_{q},$ the Gauss sum over $\mathbb{F}_{q}$ is
defined by
\begin{eqnarray}\label{eqn-MC}
G(\psi,\chi)=\sum_{ x \in \mathbb{F}^{\times}_{q}
}\psi(x)\chi(x)
\end{eqnarray}
\begin{lemma}\label{lem-Gauss}
Let $\psi$ and $\chi$ be an additive and multiplicative character
of $\mathbb{F}_{q}$ respectively. Then
$$
G(\psi,\chi) =  \left \{
\begin{array}{ll}
q-1, & \mbox{if} \  \psi =1 \ \mbox{and} \ \chi=1; \\
-1, & \mbox{if} \ \psi \neq 1 \ \mbox{and} \ \chi =1; \\ 0, &
\mbox{if} \ \psi =1 \ \mbox{and} \  \chi \neq 1.
\end{array}
\right.
$$

If $ \psi = \psi_{b} \neq 1 $ (namely, $ b \neq 0$) and $\chi
\neq 1,$ then
$$
|G(\psi,\chi)|=\sqrt{q},
$$
and
$$
G(\psi_{b},\chi)=\bar{\chi}(b)G(\chi),
$$
where
$$
G(\chi)=G(\psi_{1},\chi)=\sum_{x \in
\mathbb{F}^{\times}_{q}}\psi_{1}(x)\chi(x)=\sum_{x \in
\mathbb{F}^{\times}_{q}}\chi(x)\zeta^{T(x)}_{p}.
$$ \end{lemma}

The following theorem describes a infinite series of unit time-phase signal sets 
for the case $M=1$. 

\begin{theorem}\label{thm-CONS1}
Let $q=p^{l}, n=q-1, \mathrm{T} : \mathbb{F}_{q} \rightarrow
\mathbb{F}_{p}$ be the trace mapping, $\gamma$ be a primitive
element of $\mathbb{F}_{q}$. Let
$$
\phi=\frac{1}{\sqrt{n}}(\phi(0),\phi(1),\cdots, \phi(n-1)) \in  \C^n
$$ where
$$
\phi(i)=\zeta^{T(\gamma^{i})}_{p} \ \ (0 \leq i \leq n-1).
$$
Then $\fS = \{ \phi \}$ is an $(n,1,\frac{\sqrt{n+1}}{n})$ unit
time-phase signal set. \end{theorem}

\begin{proof}
For $(w,\tau)\neq (0,0), \ 0 \leq w,\tau \leq n-1$
\begin{eqnarray}\label{eqn-con1}
\langle \phi, \mM_{w} \lL_{\tau}(\phi) \rangle
&=&\frac{1}{n}\sum^{n-1}_{i=0}\zeta_{p}^{T(\gamma^{i})}
\bar{\zeta}_{p}^{T(\gamma^{i+\tau})}\bar{\zeta}^{iw}_{n} \nonumber \\
&=&\frac{1}{n}\sum^{n-1}_{i=0}
\zeta_{p}^{T(\gamma^{i}(1-\gamma^{\tau})})\bar{\zeta}_{n}^{iw}
\nonumber
\\ &=& \frac{1}{n}\sum_{x \in
\mathbb{F}^{\times}_{q}}\psi_{b}(x)\chi^{w}(x)=\frac{1}{n}G(\psi_{b},\chi^{w})
\end{eqnarray}
where $b=1-\gamma^{\tau}$ and $\psi_{b}$ is the additive
character of $\mathbb{F}_{q}$ defined by (\ref{eqn-AC}), $\chi$ is
the multiplicative character of $\mathbb{F}^{\times}_{q}$ defined by
$\chi(\gamma)=\bar{\zeta}_{n}$ and $G(\psi_{b},\chi^{w})$ is the
Gauss sum defined by (\ref{eqn-MC}). From Lemma \ref{lem-Gauss} we
have
$$
| \langle \phi, \mM_{w} \lL_{\tau}(\phi) \rangle |= \left \{
\begin{array}{ll}
\frac{1}{n},  & \mbox{if } \ \tau=0 \ (\mbox{so  that}\ b=0) \
\mbox{and} \ w \neq 0,
\\
0, & \mbox{if} \ \tau \neq 0 \ (\mbox{so  that} \ b \neq 0) \ \mbox{and}  \ w=0, \\
\frac{1}{n}|G(\psi_{b},\chi^{w})|=\frac{\sqrt{n+1}}{n},  & \mbox{if
}\ \tau \neq 0 \ \mbox{and} \ w \neq 0.
\end{array}
\right.
$$
Therefore $\lambda=\frac{\sqrt{n+1}}{n}.$
\end{proof}

\begin{remark} 
In the case that $M=1$, the  bound of (\ref{eqn-welchb2-DFF}) is $1/ \sqrt{n+1}.$ 
The existence of an $(n,M,\lambda)=(n,1,1/ \sqrt{n+1})$ unit time-phase signal 
set for any $n \ge 2$ is equivalent to a particular kind of SIC-POVM in quantum 
information theory (\cite{RBSC,Zaun}). But so far only for finitely many  of $n$ such 
a SIC-POVM has been constructed \cite{SGr}. In other words, no infinite series of 
$(n,1,1/ \sqrt{n+1})$ unit time-phase signal sets are known. The infinite series of 
unit time-phase signal sets of Theorem \ref{thm-CONS1} almost meet  the  bound 
of (\ref{eqn-welchb2-DFF}). 

When $q\ge 3$ and $\frac{n-1}{n^2} \le \lambda^2 \le \frac{2n^2-5n+4}{n^2(n-1)}$, 
the second bound of  Theorem \ref{thm-complexq>2} becomes 
$$ 
M \le \left\lfloor \frac{n(1-\lambda^2)}{2n-1-n^2 \lambda^2} \right\rfloor.
$$ 
It is easily verified that the infinite series of unit time-phase signal sets of Theorem 
\ref{thm-CONS1} meet this bound, and are thus optimal. This is the first time that 
an infinite series of optimal time-phase signal sets is constructed. 
\end{remark} 

From now on we assume $M \geq 2.$ In this case it is easy to see
that the lower bound
$w_{2}=(\frac{2Mn-(n+1)}{(n+1)(Mn^2-1)})^{1/4}$ is tighter than
$w_{1}=1/ \sqrt{n+1}$ and the bound
$\sqrt{\frac{2nM-n-1}{(n+1)(nM-1)}}$ is tighter than $w_{2}$ for all
$M \geq 2.$

Now we present a cyclotomic construction of unit time-phase
signal sets for $M \geq 2.$ The construction is a generalization of
Theorem \ref{thm-CONS1}.

\begin{theorem}\label{thm-CONS2}
Let $ \ q=p^{l}, \  q-1=en \ (e\geq 2), \  T: \mathbb{F}_{q}
\rightarrow \mathbb{F}_{p}$ be the trace mapping, $\gamma$ be a
primitive element of $\mathbb{F}_{q}$. For $ 0 \leq i \leq e-1,$ let
$$
\phi_{i}=\frac{1}{\sqrt{n}}(\phi_{i}(0),\phi_{i}(1),\cdots,\phi_{i}(n-1))
\in \mathbb{C} ^{n}
$$ where
$$
\phi_{i}(l)=\zeta^{T(\gamma^{i+le})}_{p} \ \ (0 \leq l \leq n-1).
$$
Then $\fS = \{ \phi_{i}: 0 \leq i \leq e-1 \}$ is an $(n, M,
 \lambda)$ unit time-phase signal set where $n= \frac{q-1}{e},M=e $
and $\lambda \leq \frac{\sqrt{en+1}}{n}.$
\end{theorem}

\begin{proof}
For $0 \leq i,j,w,\tau \leq n-1, \ (i-j,w,\tau) \neq (0,0,0) ,$
\begin{eqnarray}\label{eqn-con2}
\langle \phi_{i}, \mM_w \lL_\tau (\phi_{j}) \rangle &=& \frac{1}{n}
\sum^{n-1}_{l=0} \zeta_{p}^{T(\gamma^{i+le})}
\bar{ \zeta }_{p}^{T(\gamma^{j+(l+\tau )e})}  \bar{ \zeta }_{n} ^{lw} \nonumber \\
&=&\frac{1}{n}\sum^{n-1}_{l=0}
\zeta_{p}^{T(\gamma^{le}(\gamma^{i}-\gamma^{j+\tau
e}))}\bar{\zeta}_{n}^{lw} \nonumber
\\ &=& \frac{1}{n}\sum^{n-1}_{l=0}\zeta^{T(\beta \gamma^{le})}_{p}
\chi^{w}(\gamma^{le})
\end{eqnarray}
where $\chi$ is the multiplicative character of
$\mathbb{F}^{\times}_{q}$ defined by
$\chi(\gamma)=\bar{\zeta}_{q-1}$ and
$\beta=\gamma^{i}-\gamma^{j+\tau e}.$ Since for $\gamma^{t}, \ 0
\leq t \leq q-2,$
$$
\sum^{e-1} _{s=0}\chi^{ns}(\gamma^{t})=\sum^{e-1}
_{s=0}\bar{\zeta}_{e}^{ts} = \left \{
\begin{array}{l}
e, \ \ \ \ \ \ if \  e \mid t ;
\\
0, \ \ \ \ otherwise.
\end{array}
\right.$$
Therefore
\begin{eqnarray}\label{eqn-con3}
\frac{1}{n}\sum^{n-1}_{l=0}\zeta^{T(\beta \gamma^{le})}_{p}
\chi^{w}(\gamma^{le}) &=& \frac{1}{en} \sum_{x
\in\mathbb{F}^{\times}_{q}}\zeta^{T(\beta
 x)}_{p}\chi^{w}(x)\sum^{e-1}_{s=0}\chi^{ns}(x) \nonumber \\
&=& \frac{1}{en}\sum^{ e - 1}_{s = 0}\sum_{x \in
\mathbb{F}^{\times}_{q}}\chi^{ns+w}(x)\zeta^{T(\beta
 x)}_{p}.
\end{eqnarray}

If $(j-i,\tau)=(0,0),$ we have $1 \leq w \leq n-1, \ \beta = \gamma^{i}
- \gamma^{ j + \tau e} =0 $ and $ \chi^{ns + w} \neq 1$ for all $s \
(0 \leq s \leq e-1).$ Therefore by (\ref{eqn-con2}) and
(\ref{eqn-con3}),
$$
\langle \phi_{i},\mM_w \lL_\tau (\phi_{j}) \rangle
=\frac{1}{en}\sum^{e-1}_{s=0}\sum_{x \in
\mathbb{F}^{\times}_{q}}\chi^{ns+w}(x)=\frac{1}{en}\sum^{e-1}_{s=0}0=0.
$$
If $(j-i,\tau) \neq (0,0),$ then $\beta \neq 0$ and the right hand
side of (\ref{eqn-con3}) is
$\frac{1}{en}\sum^{e-1}_{s=0}\bar{\chi}^{ns+w}(\beta)G(\bar{\chi}^{ns+w}).$
Therefore
$$
\langle \phi_{i},  \mM_w \lL_\tau  (\phi_{j}) \rangle \leq
\frac{1}{en}\sum^{e-1}_{s=0}|G(\bar{\chi}^{ns+w})| \leq \frac{e
\sqrt{q} }{en}=\frac{\sqrt{en+1}}{n}.
$$
The upper bound on $\lambda$ then follows.
\end{proof}

\begin{remark}
For the $(n,M)=((q-1)/e, e)$ time-phase signal set $\fS$ in Theorem \ref{thm-CONS2}, the bound of
Theorem \ref{thm-complexq>2} is
$ 
\sqrt{\frac{2ne-e-n}{n^2e-n}}, 
$
which is very close to $\sqrt{\frac{en+1}{n^2}} \ge \lambda_\fS$ when $e$ is mall. 
Note that we were not able to compute the exact value $\lambda_\cS$ for the 
time-phase signal set $\fS$ in Theorem \ref{thm-CONS2}, and thus unable to  
tell if it is optimal with respect to some of the bounds described in this paper. 
\end{remark}

\section{Summary and concluding remarks}\label{sec-summary}

In this paper, we developed a number of bounds on unit time-phase signal sets
which are derived from existing bounds on unit time signal sets. Although the techniques
used in establishing Theorems \ref{thm-mmain}  and \ref{thm-mmain2} are simple,
they are very useful for developing better bounds on unit time-phase signal sets. These two
techniques employ the two one-way bridges with which a unit time-phase signal set
$\cS$ is converted into the two unit time signal sets $\cS_{\cS}$ of
(\ref{eqn-mainS}) and $\bar{\cS}_{\cS}$ of (\ref{eqn-mainS2}). If one sees these
bridges, one would immediately obtain the new bounds on unit time-phase signal sets.
However, without seeing
these bridges, it may be hard to develop bounds on time-phase signal sets even
if one is very familiar with the Welch bound and Levenstein bounds.
Theorems
\ref{thm-mmain}  and \ref{thm-mmain2} are generic and can be employed
to obtain more bounds on unit time-phase signal sets from new bounds on unit
time signal sets. Thus, one of the main contributions of this paper is the discovery
of these two one-way bridges and Theorems
\ref{thm-mmain}  and \ref{thm-mmain2}.

The two one-way bridges described in (\ref{eqn-mainS}) and (\ref{eqn-mainS2}) also
suggest two ways to construct good time signal sets from a good time-phase
signal set. However, it is open how to construct a good time-phase signal set
from a given good time signal set.

It is noticed that time-phase signal sets and time signal sets (also called codebooks)
are very different,
though they all are subsets of $\hH_n$. This is because time-phase signal sets
consider both the time and phase distortion, while time signal sets take care of
only the time distortion. It is much harder to construct good time-phase signal
sets. So far no optimal $(n,M)$ time-phase signal set with $M>1$ is known, while a number of optimal
time signal sets (codebooks) have been constructed in the
literature \cite{CCKS,Ding05,DF,DYin,Kary1,XZG}.

An interesting problem is to construct unit time-phase signal sets meeting or
almost meeting the bounds on unit time-phase signal sets described in this 
paper if this is possible. The infinite series of 
unit time-phase signal sets of Theorem \ref{thm-CONS1} are  optimal. This 
is the first time that an infinite series of optimal time-phase signal sets are 
constructed. The time-phase signal sets of Theorem \ref{thm-CONS2} are 
also very good when $e$ 
is small. The 
constructions of these optimal and almost optimal time-phase signal sets are 
another major contribution of this paper. In general, the bounds on unit 
time-phase signal sets described in this paper should be very good as they 
are derived from the linear programming bounds on unit time signal sets.

Given any $(n, 1, \lambda)$ unit time-phase signal set $\fS$ meeting the bound
of   (\ref{eqn-welchb2-DFF}) (meeting also the bound of (\ref{eqn-LDFF2}) since
$M=1$),  the set  $\cS_{\fS}$ defined in (\ref{eqn-mainS}) is an $(n, n^2, 1/\sqrt{n+1})$
signal set meeting the Welch bound. Such $(n, n^2, 1/\sqrt{n+1})$ signal sets are called
{\em SIC-POVMs} in quantum information \cite{Zaun,Gras,RBSC}. Algebraic and numerical
constructions of such $(n, n^2, 1/\sqrt{n+1})$ signal sets are known for small dimensions
$n$ \cite{Gras,SGr}. It is conjectured that SIC-POVMs exist for every dimension $n$.
However, constructing SIC-POVMs seems to be a very hard problem. Therefore, it
is also a hard problem to construct
$(n, 1, \lambda)$ unit time-phase signal set $\fS$ meeting the bound of
(\ref{eqn-welchb2-DFF}). SIC-POVMs may be used to construct $(n, 1, \lambda)$
unit time-phase signal set $\fS$ meeting the bound of  (\ref{eqn-welchb2-DFF}).
It would be worthy to investigate this. Optimal  $(n, n^2, 1/\sqrt{n+1})$ time signal sets are also related to mutually unbiased
bases, tight frames and line packing in Grassmannian space \cite{CHS,BCM,SH03, DYin,CCKS,Ding05}.

Incoherent systems are also related to time signal sets \cite{HerSto,NT11}. Bounds on
incoherent systems may be employed to derive bounds on time-phase signal sets with
the two bridges established in this paper.

\section*{Acknowledgements}

This work was done while the authors were attending the  Summer School for
Mathematical Foundations of Coding Theory and Cryptography hosted by the
Beijing International Center for Mathematical Research. The authors wish to
thank the Center for its support in many aspects.

C. Ding is supported by the Hong Kong Research Grants Council under Project No. 601311.
K. Feng and R. Feng are supported by the NSFC Grant No. 10990011.
K. Feng is also supported by the Tsinghua National Lab. for
Information Science and Technology.

\end{document}